\documentclass[conference]{IEEEtran}

\usepackage[dvips]{color}
\usepackage{comment}
\usepackage{epsf}
\usepackage{times}
\usepackage{epsfig}
\usepackage{graphicx}
\usepackage{amsmath}
\usepackage{amssymb}
\usepackage{amsxtra}
\usepackage{here}
\usepackage{rawfonts}
\usepackage{times}
\usepackage{url}
\usepackage{cite}
\usepackage{amssymb}
\usepackage{amsmath}
\usepackage[dvips]{color}
\usepackage{epsf}
\usepackage{times}
\usepackage{epsfig}
\usepackage{graphicx}
\usepackage{amsmath}
\usepackage{amssymb}
\usepackage{amsxtra}
\usepackage{here}
\usepackage{rawfonts}
\usepackage{times}
\usepackage{url}
\usepackage{cite}
\usepackage{multirow}

\usepackage{epstopdf}

\usepackage{array}
\usepackage{amsmath,epsfig,amssymb,algorithm,algpseudocode,amsthm,cite,url}
\usepackage{here}
\usepackage{tabu}
\usepackage{geometry}
\usepackage{caption}
\usepackage{amsmath}

\usepackage[justification=centering]{caption}
\usepackage{verbatim}

\newtheorem{proposition}{\bf Proposition}

\newlength{\aligntop}
\newlength{\alignbot}
\makeatletter
\renewenvironment{align}{%
  \vspace{\aligntop}
  \start@align\@ne\st@rredfalse\m@ne
}{%
  \math@cr \black@\totwidth@
  \egroup
  \ifingather@
    \restorealignstate@
    \egroup
    \nonumber
    \ifnum0=`{\fi\iffalse}\fi
  \else
    $$%
  \fi
  \ignorespacesafterend%
  \vspace{\alignbot}\par\noindent
}
\IEEEoverridecommandlockouts
\IEEEoverridecommandlockouts
\begin{document}
\title{\huge Drone Small Cells in the Clouds: Design, Deployment and Performance Analysis\vspace{-0.4cm}}

\author{\authorblockN{ Mohammad Mozaffari, Walid Saad, Mehdi Bennis and Merouane Debbah} \authorblockA{\small
Wireless@VT, Electrical and Computer Engineering Department, Virginia Tech, VA, USA, Emails: \url{{mmozaff , walids}@vt.deu}\\
 }

 }
\author{\authorblockN{ Mohammad Mozaffari$^1$, Walid Saad$^1$, Mehdi Bennis$^2$, and Merouane Debbah$^3$} \authorblockA{\small
$^1$ Wireless@VT, Electrical and Computer Engineering Department, Virginia Tech, VA, USA, Emails: \url{{mmozaff , walids}@vt.deu}\\
$^2$ CWC - Centre for Wireless Communications, Oulu, Finland, Email: \url{bennis@ee.oulu.fi}\\
$^3$ Mathematical and Algorithmic Sciences Lab, Huawei France R \& D, Paris, France, Email:\url{merouane.debbah@huawei.com}
}%
\thanks{This research was supported by the U.S. National Science Foundation under Grant AST-1506297.}
}

%\author{Mohammad Mozaffari, Walid Saad, Mehdi Bennis and Merouane Debbah}
%$^2$ Electrical and Computer Engineering Department, University of Houston, Houston, TX, USA, Email: \url{zhan2@mail.uh.edu}\\
%$^3$ Electrical Engineering Department, Princeton University, Princeton, NJ, USA, Emails: \url{poor@princeton.edu}\\
%$^4$ Coordinated Science Laboratory, University of Illinois at Urbana-Champaign, IL, USA, E-mail: \url{basar1@illinois.edu}\vspace{-1cm}
% }%
%   \thanks{This material is based upon work supported by the DTRA under Grant No. HDTRA1-07-1-0037 and is also supported by NSF projects CNS-0910461, CNS-0953377, CNS-0905556, and ECCS-1028782.}
% }
%\date{}
\maketitle

\begin{abstract}
The use of drone small cells (DSCs) which are aerial wireless base stations that can be mounted on flying devices such as unmanned aerial vehicles (UAVs), is emerging as an effective technique for providing wireless services to ground users in a variety of scenarios. The efficient deployment of such DSCs while optimizing the covered area is one of the key design challenges. In this paper, considering the low altitude platform (LAP), the downlink coverage performance of DSCs is investigated. The optimal DSC altitude which leads to a maximum ground coverage and minimum required transmit power for a single DSC is derived. Furthermore, the problem of providing a maximum coverage for a certain geographical area using two DSCs is investigated in two scenarios; interference free and full interference between DSCs. The impact of the distance between DSCs on the coverage area is studied and the optimal distance between DSCs resulting in maximum coverage is derived. Numerical results verify our analytical results on the existence of optimal DSCs altitude/separation distance and provide insights on the optimal deployment of DSCs to supplement wireless network coverage.\vspace{-0.01cm}
\end{abstract}

\section{Introduction}\vspace{-0.1cm}
Recently, using aerial base stations to support ground cellular networks has received significant attention. Particularly,  drone small cells (DSCs) can act as aerial base stations to support cellular networks in high demand and overloaded situations, or for the purpose of public safety and disaster management \cite{R1}. The main advantage of using DSCs is that they do not need to have an actual pilot and hence they can be autonomously deployed in dangerous environments for the purpose of search, rescue and communication. Furthermore, since DSCs are essentially mobile base stations, they are more robust against environmental changes as compared to fixed ground base stations. For example, in disaster situations such as earthquakes or floods where some of ground base stations are damaged, or during big public events such as Olympic Games where a huge demand for communication is observed, the cellular network needs to be assisted to provide the needed capacity and coverage \cite{R1}. In these cases, deploying DSCs acting as base stations is extremely useful in providing an improved quality-of-service (QoS) for ground users. The deployment of DSCs faces many challenges such as power consumption, coverage optimization and interference management \cite{R2}.

To address some of these challenges, the authors in \cite{R2} provided a general view of practical considerations for the integration of DSCs with cellular networks. The work in \cite{Overload}, considered the use of DSCs to compensate for the cell overload and outage in cellular networks. However, in this body of work there is no extensive discussion on the coverage performance of DSCs and deployment methods. Due to the special application of DSCs in unexpected events such as disaster, rapid and efficient deployment of DSCs is mandatory to provide a maximum coverage service for ground users. DSCs can be deployed in a high altitude platform (HAP) which is above 10 km height or in low altitude platform (LAP) with the altitude of below 10 km \cite{R3}. In \cite{R4} the optimal deployment/movement of DSCs in order to improve the connectivity of  wireless ad-hoc networks was studied. In \cite{R5}, considering  static ground users, the optimum trajectory and heading of DSCs equipped with multiple antennas for ground to air uplink scenario was investigated.\

Beyond deployment, another important challenge for mobile DSC base stations is channel modeling. For instance, in \cite{R6}, the probability of line of sight (LOS) for air to ground communication as a function of elevation angle and average height of buildings in a dense urban area was determined. The air to ground path loss model has been presented in \cite{R7},\cite{R8}. As discussed in \cite{R8} , due to path loss and shadowing effects of obstacles, the characteristics of the air to ground channel depend on the height of the aerial base stations. By increasing the altitude of a DSC, the path loss increases, however, shadowing effect decreases and the possibility of having LOS connections between ground users and DSCs increases. Therefore, an optimum altitude for the aerial base station which results in a maximum coverage exists. In \cite{R10}, assuming only one DSC operating with no inter-cell interference, the optimum altitude for the DSC which leads to a maximum coverage was derived. However, the authors did not consider the case of multiple DSCs where beyond altitude, the distance between DSCs also impacts the overall coverage performance. The problem of multiple DSC deployment becomes even more challenging when  interference occurs between the received signal from different DSCs. The impact of interference on the coverage performance of DSC has not been investigated in prior studies.

The main contribution of this paper is to develop fundamental results on the coverage and optimal deployment of wireless DSCs. First, we analyze the optimal height for a DSC for which the required transmit power for covering a target area is minimized. Next, to achieve the maximum coverage performance for a specified area, the optimal deployment of two DSCs in both interference and interference-free situations is studied. The goal is to maximize the coverage performance of DSCs by calculating optimal values for DSCs altitude and the distance between them. To this end, we consider a target area with a specific size and for a single static DSC and we find the optimum altitude ensuring sufficient coverage using minimum transmit power. Next, with the goal of offering maximum coverage for the target area the optimal deployment of two DSCs over the area is investigated. Numerical evaluations are then used to validate the derived analytical results.\

The rest of this paper is organized as follows: Section II presents the system model describing the air to ground channel model. In Section III, coverage performance of a single DSC and multiple DSCs is investigated. In Section IV, we present the numerical results. Finally, Section V concludes the paper.

\section{System Model and the Single DSC Case}
Consider a static DSC located at an altitude of $h$ meters transmitting signals to static users on the ground. In order to analyze the coverage of such a DSC, it is imperative to adopt an appropriate path loss model that is suitable for air to ground communication. In this section after presenting the air to ground channel model, the optimal altitude for a single DSC case is derived.
\subsection{Air to Ground Channel Model}\label{sec:sysmodel}
As discussed in \cite{R3} and \cite{R9}, the ground receiver receives three groups of signals including LOS, strong reflected signals (NLOS) and multiple reflected components which cause multipath fading. These groups can be considered separately with different probabilities of occurrence. Typical, as discussed in \cite{R7}, it is assumed that the received signal is categorized only in one of the mentioned groups. Each group has a specific probability of occurrence which is a function of environment, density and height of buildings and elevation angle. The probability of receiving LOS and strong NLOS components are significantly higher than fading \cite{R7}. Therefore, the impact of small scale fading can be neglected. A common approach to model air to ground propagation channel is to consider LOS and NLOS components along with their occurrence probabilities separately. Note that for NLOS connections due to the shadowing effect and reflection of signals from obstacles, path loss is higher than LOS. Hence, in addition to the free space propagation loss, different excessive path loss values are assigned to LOS and NLOS links.

\ Figure 1 shows a DSC located at an altitude of $h$ and ground users at the radius of $R$ from a point corresponding to the projection of DSC onto the ground. The distance between the DSC and the ground receiver is $d = \sqrt {{R^2} + {h^2}}$ and $\theta  = {\tan ^{ - 1}}(h/R)$ indicates the elevation angle (in radian) DSC with respect to the user.

The path loss for LOS and NLOS connections are \cite{R3}:\vspace{-0.3cm}

\begin{align}
{L_{{\text{LoS}}}}(dB) = 20\log (\frac{{4\pi {f_c}d}}{c}) + {\xi _{{\text{LoS}}}},\\
{L_{{\text{NLoS}}}}(dB) = 20\log (\frac{{4\pi {f_c}d}}{c}) + {\xi _{{\text{NLoS}}}},
\end{align}
where ${L_{{\text{LoS}}}}$ and ${L_{{\text{NLoS}}}}$ are the average path loss for LOS and NLOS links, ${\xi _{{\text{LoS}}}}$ and ${\xi _{{\text{NLoS}}}}$ are the average additional loss to the free space propagation loss which depend on the environment, $c$ is the speed of light, ${f_{\text{c}}}$ is the carrier frequency and $d$  is the distance between the DSC and ground receiver.
\begin{figure}[!t]
  \begin{center}
   \vspace{-0.2cm}
    \includegraphics[width=2.5cm]{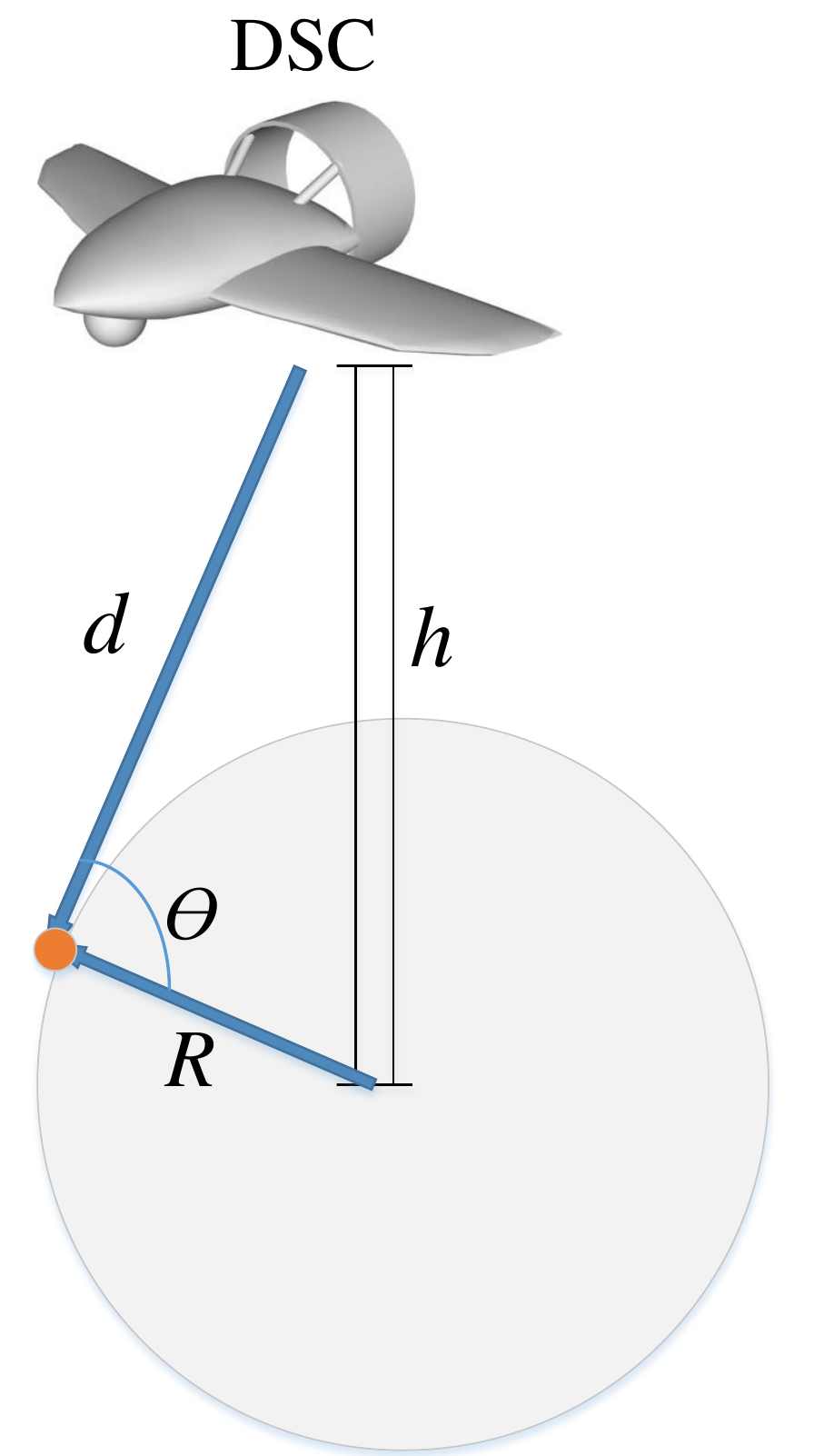}
    \vspace{-0.1cm}
    \caption{\label{fig: Fig1} Low altitude platform.\vspace{-0.5cm}}
  \end{center}\vspace{-0.3cm}
\end{figure}

The probability of having LOS connections at an elevation angle of  $\theta $ is given by \cite{R10}
\begin{equation}\label{LOS}
{\text{P}}(\text{LOS}) = \frac{1}{{1 + \alpha \exp ( - \beta \left[ {\tfrac{{180}}{\pi }\theta  - \alpha } \right])}} ,
\end{equation}
where $\alpha$  and $\beta$  are constant values which depend on the environment (rural, urban, dense urban,  etc.). Also, probability of NLOS is  ${\text{P}}({\text{NLOS}}) = 1 - {\text{P}}({\text{LOS}})$. Equation (\ref{LOS}) indicates that the probability of having LOS connection between the aerial base station and to ground users is an increasing function of elevation angle. In other words, by increasing the elevation angle between the receiver and transmitter, the shadowing effect decreases and clear LOS path exists with high probability. Finally, the average path loss as a function of the DSC altitude and coverage radius becomes
\begin{equation}\label{PL}
\overline L (R,h) = {\text{P}}({\text{LoS}}) \times {L_{\text{LoS}}} + {\text{P}}({\text{NLOS}}) \times {L_{\text{NLoS}}}.
\end{equation}\vspace{-0.5cm}

\subsection{Optimal Altitude for Single DSC }\vspace{-0.1cm}
Given this channel model, our first goal is to study the problem of optimal altitude for a single DSC seeking maximum ground coverage.
Consider a DSC transmitting its signal with the power of ${P_{\text{t}}}$, then the received power is written as
\begin{equation}\label{eq2:solve}
{P_{\text{r}}}(dB) = {P_{\text{t}}} - \overline L (R,h).
\end{equation}

A point on the ground is covered by the DSC if it’s signal to noise ratio (SNR) is greater than a threshold (${\gamma _{{\text{th}}}}$ ). That is
\begin{equation}
\gamma (R,h) = \frac{{{P_{\text{r}}}}}{N} \geqslant {\gamma _{{\text{th}}}},
\end{equation}
where $N$ is the noise power.
Obviously, to find the maximum achievable coverage radius we should have $\gamma (R,h) = {\gamma _{{\text{th}}}}$. For a fixed transmit power, the optimal DSC height which results in maximum coverage is computed by solving the following equation \cite{R7}:
\begin{equation}\label{eq1:solve}
\frac{{180({\xi _{{\text{NLoS}}}} - {\xi _{{\text{LoS}}}})\beta Z}}{{\pi {{(Z + 1)}^2}}} - \frac{{20\mu }}{{\log (10)}} = 0,
\end{equation}
where $Z = \alpha \exp ( - \beta \left[ {\tfrac{{180}}{\pi }{{\tan }^{ - 1}}(\mu ) - \alpha } \right])$ and $\mu  = h/R$.
By solving  (\ref{eq1:solve}), ${\mu _{{\text{opt}}}} = {h_{{\text{opt}}}}/{R_{\max }}$
 is computed and using (\ref{eq2:solve}), ${h_{{\text{opt}}}}$ and ${R_{\max }}$
  are found.

Note that due to the limitation on the altitude of DSCs, we have $h \leqslant {h_{\max }}$, where ${h_{\max }}$ is the maximum allowable altitude for DSCs. It can be shown that using the typical values for the parameters in  (\ref{PL}) and (\ref{LOS}), $\frac{{{\partial ^2}h}}{{{\partial ^2}R}} < 0$ which implies that $R$ as a function of $h$ is a concave function. Therefore,  the coverage range increases as the altitude increases up to the optimal point and after that it decreases. As a result, considering a constraint on the maximum allowable altitude, the feasible optimal altitude is equal to ${\hat h_{{\text{opt}}}} = {\text{min}}\{ {h_{\max }},{h_{{\text{opt}}}}\}$.
Now, assume that the target area which should be covered is fixed with radius of ${R_{\text{c}}}$
 and the goal is to find the an optimal altitude where the minimum transmit power is required to cover the target area. The derivative of transmit power with respect to the altitude is:
\begin{equation}
\partial {P_{\text{t}}}/\partial h = \partial \overline L ({R_{\text{c}}},h)/\partial h = 0 \to {h_{{\text{opt}}}} = {\mu _{{\text{opt}}}}{R_{\text{c}}}.
\end{equation}

Finally, considering the feasible optimal altitude, the minimum required transmit power will be
\begin{equation}
{P_{{\text{t,min}}}}(dB) = \overline L ({R_{\text{c}}},{\hat h_{{\text{opt}}}}) + {\gamma _{{\text{th}}}}N.
\end{equation}\vspace{-0.45cm}

Now, we prove that $R$ as a function of $h$ does not have more than one local maximum. In other words, if a local maximum exists, the corresponding $h$ is the optimal altitude. Clearly, if a DSC is deployed at the optimal altitude, it provides a maximum SNR for any ground users. This is equivalent to have a minimum path loss for the users. Consider a ground user located at the radius of ${R_o}$ from a point corresponding to the projection of a DSC onto to the ground. The average path loss at the user location as a function of elevation angle can be written as:

\begin{equation}
\begin{aligned}
\overline L (\theta ) = \frac{{({\xi _{{\text{LoS}}}} - {\xi _{{\text{NLoS}}}})}}{{1 + \alpha \exp ( - \beta \left[ {\tfrac{{180}}{\pi }\theta  - \alpha } \right])}} \\- 20\log ({R_o}cos(\theta )) + 20\log (\frac{{4\pi {f_c}d}}{c}).
\end{aligned}
\end{equation}

Since altitude and elevation angle are directly related, the optimal altitude corresponds to the optimal elevation angle. To show that the number of local minimum path loss as a function of elevation angle is not greater than one, we should have:
\begin{proposition}
%\begin{equation}
If a local minima exists in the path loss function, then it is the only local minima of the function.
%\end{equation}
\end{proposition}
\begin{proof}
\begin{gather*}
\begin{aligned}
&{\text{we have to show if}} {\text{       }}\frac{{\partial \overline L (\theta )}}{{\partial \theta }} > 0{\text{ }} \to \frac{{{\partial ^2}\overline L (\theta )}}{{\partial {\theta ^2}}} > 0.\\
&  \frac{{\partial \overline L (\theta )}}{{\partial \theta }} = \frac{{\tfrac{{180}}{\pi }\beta ({\xi _{{\text{NLoS}}}} - {\xi _{{\text{LoS}}}})Z}}{{{{(1 + Z)}^2}}} + \tan (\theta ) \hfill \\
&  \frac{{\partial \overline L (\theta )}}{{\partial \theta }} > 0 \to \tan (\theta ) > \frac{{\tfrac{{180}}{\pi }\beta ({\xi _{{\text{NLoS}}}} - {\xi _{{\text{LoS}}}})Z}}{{{{(1 + Z)}^2}}}, \hfill \\
\end{aligned}
\end{gather*}
\raggedright \text{So,}
\begin{gather*}
\begin{aligned}
 & {\left[ {\tan (\theta )} \right]^2} > \frac{{{{\left[ {\tfrac{{180}}{\pi }\beta ({\xi _{{\text{LoS}}}} - {\xi _{{\text{NLoS}}}})Z} \right]}^2}}}{{{{(1 + Z)}^4}}} \hfill \\
 &  = \frac{{{{\left[ {\tfrac{{180}}{\pi }\beta } \right]}^2}({\xi _{{\text{LoS}}}} - {\xi _{{\text{NLoS}}}}){Z^3}}}{{{{(1 + Z)}^4}}} \times \frac{{({\xi _{{\text{LoS}}}} - {\xi _{{\text{NLoS}}}})}}{Z} \hfill \\
&  \mathop  > \limits^{(a)} \frac{{{{\left[ {\tfrac{{180}}{\pi }\beta } \right]}^2}({\xi _{{\text{LoS}}}} - {\xi _{{\text{NLoS}}}}){Z^3}}}{{{{(1 + Z)}^4}}}, \hfill \\
\end{aligned}
\end{gather*}
\raggedright \text{Finally},
\begin{small}
\begin{gather*}
 \frac{{{\partial ^2}\overline L (\theta )}}{{\partial {\theta ^2}}} = \frac{{ - {{\left[ {\tfrac{{180}}{\pi }\beta } \right]}^2}({\xi _{{\text{LoS}}}} - {\xi _{{\text{NLoS}}}})({Z^3} - Z)}}{{{{(1 + Z)}^4}}}
+ {\tan ^2}(\theta ) + 1 > 0,
\end{gather*}
\end{small}
where $Z = \alpha \exp ( - \beta \left[{\theta - \alpha } \right])$ and (a) is based on ${(\xi _{{\text{LoS}}}} - {\xi _{{\text{NLoS}}}}) > Z$ which is hold for the typical values related to urban environments and elevation angles greater than 5 degree.

Now, assume $\theta  = {\theta _o}$ is a local minimum, then,
\begin{gather*}
{\left. {\frac{{\partial \overline L (\theta )}}{{\partial \theta }}} \right|_{\theta  = \theta _o^ + }} > 0 \to {\text{for }}\theta  > {\theta _o}{\text{ we have }}\frac{{{\partial ^2}\overline L (\theta )}}{{\partial {\theta ^2}}} > 0.
\end{gather*}
Hence,  $\theta  = {\theta _o}$ is the only local minimum and is the optimal elevation angle.
\end{proof}
Knowing that the path loss as a function of altitude has only one local minima, the optimal altitude can be found by increasing the DSC altitude up to a point where the path loss starts increasing.
\section{Case of Two Non-Interfering DSCs }\label{sec:Deployment}\vspace{0.1cm}
 Here, assuming that two DSCs are operating together in a given area, the optimal distance between them in both interference and interference-free situations is analyzed.

\subsection{Two DSCs in interference free situations}\label{sec:da}\vspace{-0.1cm}

Next, we consider two DSCs that are used to provide coverage for a target area. Here, without loss of generality, we consider the target area to be a rectangle whose length is given by $a$ and whose width is given by $b$. To have  maximum coverage for this target area, optimal values for DSCs altitude and distance should be determined.
Intuitively, for a given target area in the absence of interference between the two DSCs, the maximum overall coverage is obtained if the effective coverage inside the target area provided by each DSC is maximized while the overlap between the coverages of DSCs is minimum. These conditions are satisfied if each DSC is at its optimal altitude and they are  separated as far as possible but they should not cover outside the target area. In general, the DSCs can be deployed in different altitude and they might use a different transmit power. As a result, they can provide a different coverage radius.

Figure \ref{fig: Fig2} shows the coverage of two DSCs located at their optimal altitudes, $D$ is the distance between DSCs, $R_1^{\max }$ and $R_2^{\max }$ corresponds to the maximum coverage radius for the first and second DSC, and $O (x,y)$ is the origin of coverage area with respect to the center of target area. The optimal deployment of two DSCs in the absence of interference can be determined by the following set of equations:

\begin{figure}[!t]
  \begin{center}
   \vspace{-0.2cm}
    \includegraphics[width=6cm]{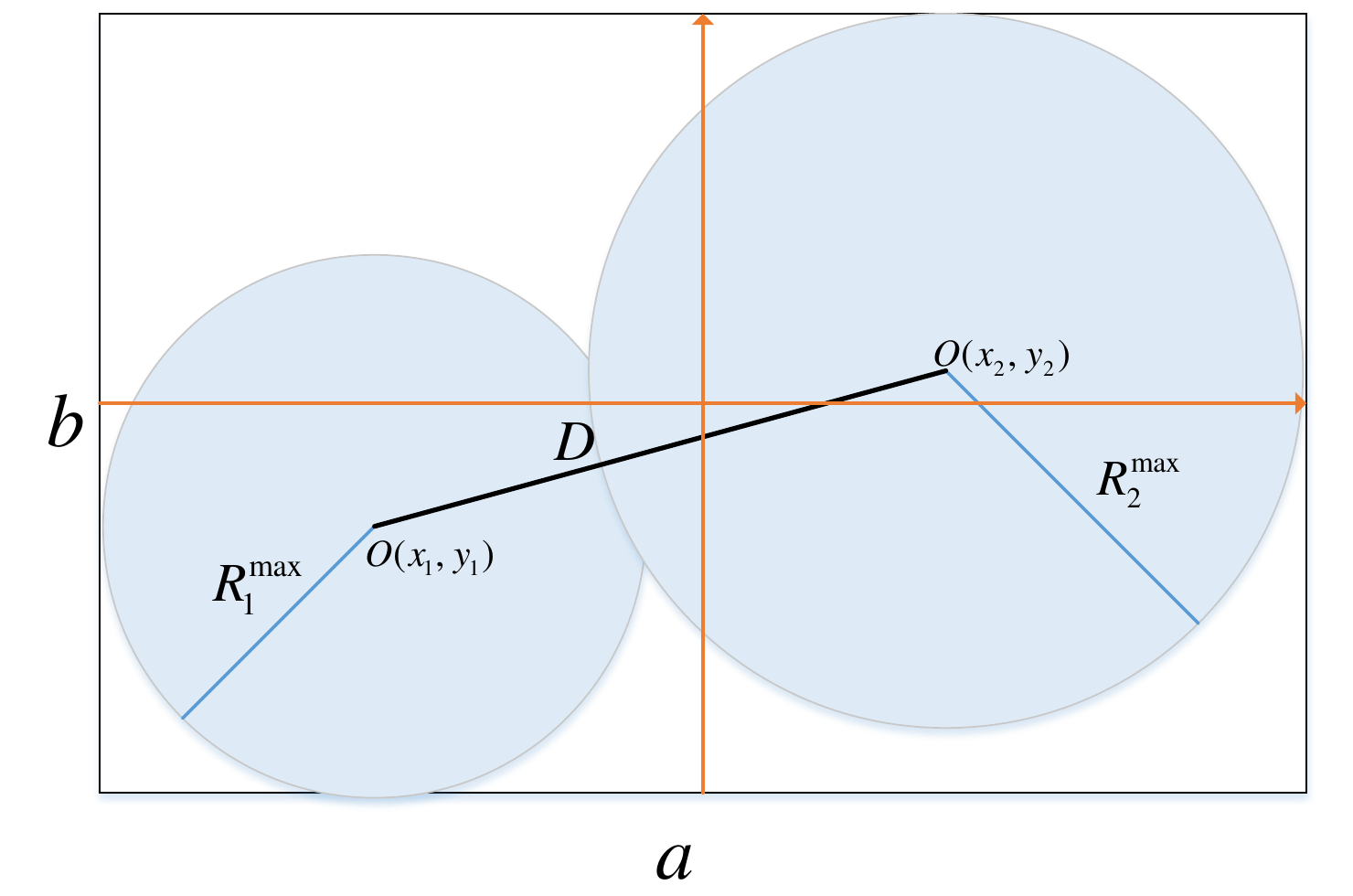}
    \vspace{-0.1cm}
    \caption{\label{fig: Fig2}  Optimal deployment of two DSCs in the absence of interference.\vspace{-1cm}}
  \end{center}
\end{figure}%\vspace{10.2cm}

\begin{equation} \label {no-inter}
\begin{cases}
  {h_1} = \hat h_1^{{\text{opt}}}, \hfill \\
  {h_2} = \hat h_2^{{\text{opt}}}, \hfill \\
  ({x_1},{y_1}) = (\frac{{ - a}}{2} + R_1^{\max },\frac{{ - b}}{2} + R_1^{\max }), \hfill \\
  ({x_2},{y_2}) = (\frac{a}{2} - R_2^{\max },\frac{b}{2} - R_2^{\max }), \hfill \\
\end{cases}
\end{equation}
where  $\hat h_1^{{\text{opt}}}$ and $\hat h_2^{{\text{opt}}}$ are the optimal feasible altitude for ${{\text{DSC}}_{\text{1}}}$ and ${{\text{DSC}}_{\text{2}}}$. (\ref{no-inter}) is found by placing the coverage areas as separate as possible and the tangent to the borders of target area. Note that in this case the target area is larger than the coverage region of UAVs and as a result the coverage regions will be located inside the target area.
By using some geometric properties for calculating the total area of intersecting circles, the maximum overall coverage area can be expressed as follows:

If the two coverage areas overlap $(D \leqslant R_1^{{\text{max}}} + R_2^{{\text{max}}})$,

%\begin{equation}
 \begin{align}
 A_{\text{C}}^{{\text{max}}}  &= \pi \big[{(R_1^{{\text{max}}})^2} + {(R_2^{{\text{max}}})^2}] \nonumber \\
& -{(R_1^{{\text{max}}})^2}{\cos ^{ - 1}}\big[\frac{{{D^2} + {{(R_1^{{\text{max}}})}^2} - {{(R_2^{{\text{max}}})}^2}}}{{2DR_1^{{\text{max}}}}}]   \nonumber \\
 &  - {(R_2^{{\text{max}}})^2}{\cos ^{ - 1}}\big[\frac{{{D^2} + {{(R_2^{{\text{max}}})}^2} - {{(R_1^{{\text{max}}})}^2}}}{{2DR_2^{{\text{max}}}}}] + B ,
 \end{align}
%\end{equation}
where \vspace{-0.5cm}\\
\begin{small}
\begin{align*}
& B = \sqrt {( - D + R_1^{\max } + R_2^{\max })(D - R_1^{\max } + R_2^{\max })}\vspace{0.2cm}\\
&\times \sqrt{(D + R_1^{\max } - R_2^{\max })(D + R_1^{\max } + R_2^{\max })}.
\end{align*}
\end{small}
For the special case where the two DSCs are identical, located at the same altitude and use the same transmit power, they have the same coverage radius ($R_1^{{\text{max}}} = R_2^{{\text{max}}} = {R^{{\text{max}}}}$). Then (13) is reduced to
%\begin{gather*}
\begin{equation}
\begin{split}
  A_C^{{\text{max}}} &= 2\pi {({R^{{\text{max}}}})^2} - 2{({R^{{\text{max}}}})^2}{\cos ^{ - 1}}\left( \frac{D}{{2{R^{{\text{max}}}}}}\right)\\
   &+ \frac{D}{2}\sqrt {4{{({R^{{\text{max}}}})}^2} - {D^2}} , \hfill \\
\end{split}
\end{equation}

%\end{equation}

If  $D > R_1^{{\text{max}}} + R_2^{{\text{max}}}$, they do not overlap and the total coverage area is given by
\begin{equation}
A_{\text{C}}^{{\text{max}}} = \pi \big[{(R_{\text{1}}^{{\text{max}}})^2} + {(R_{\text{2}}^{{\text{max}}})^2}].
\end{equation}\vspace{-0.3cm}

\subsection{Case of Two Interfering DSCs}\label{sec:inter}\vspace{0.1cm}
Next, we consider a case in which the two DSCs  interfere with each other during the transmission. This situation happens when DSCs are not controlled by the same control system so they might use the same transmit channel. Also, due to the limited number of available channels in a wireless network, the DSCs might transmit over the same channel resulting in interference.

Consider a given target area which should be covered by two DSCs. Clearly, the distance should not be too large to avoid covering unwanted area (outside the target area), and it should not be too small due to the high interference effect. Therefore, an optimum distance between DSCs which results in the highest coverage exists. Figure 3 illustrates two DSCs separated by $D$. Consider a ground user at the radius of ${R_1}$ and ${R_1}$  from the projection of ${{\text{DSC}}_{\text{1}}}$ and ${{\text{DSC}}_{\text{2}}}$ onto the ground. $\phi$ is the angle between ${\vec R_1}$ and $\vec D$. In this case, a point on the ground is covered by a DSC if the signal to interference plus noise ratio (SINR) be greater than ${\gamma _{{\text{th}}}}$. Thus \vspace{-0.5 cm}

\begin{figure}[!t]
  \begin{center}
   \vspace{-0.2cm}
    \includegraphics[width=5cm]{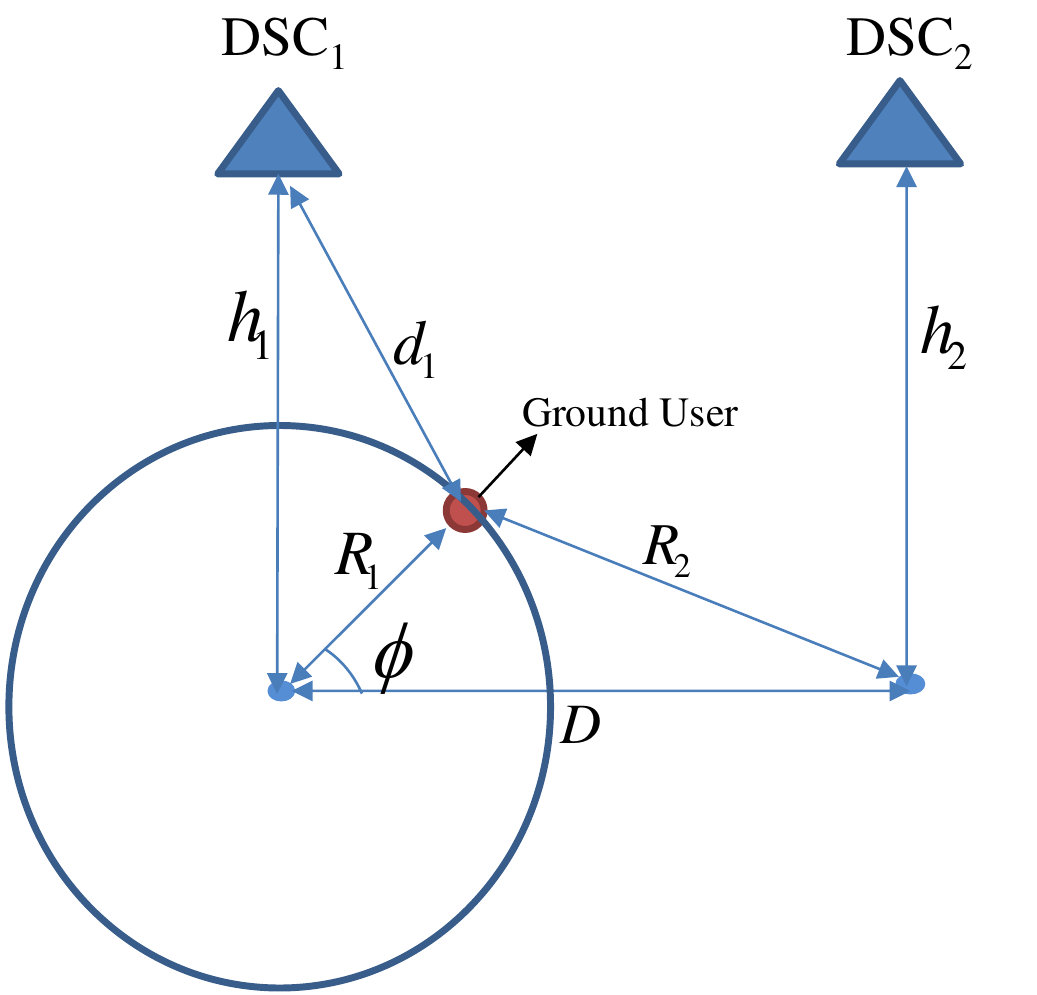}
    \vspace{-0.1cm}
    \caption{\label{fig: 2DSCs}  Two DSCs interfering scenario. \vspace{-0.5cm}}
  \end{center}\vspace{-0.2cm}
\end{figure}

\begin{equation}
\gamma ({R_1},{R_2},{h_1},{h_2}) = \frac{{{P_{{\text{r,1}}}}}}{{N + {P_{{\text{r,2}}}}}} \geqslant {\gamma _{{\text{th}}}},
\end{equation}
where ${P_{{\text{r,1}}}}$ and  ${P_{{\text{r,2}}}}$   are the received power from the first and second DSCs.

Given that ${R_2}^2 = {R_1}^2 + {D^2} - 2{R_1}D\cos (\phi )$, and assuming that the DSCs have the same altitude of $h$, the SINR can be rewritten as\vspace{-0.3cm}

\begin{equation}
\gamma ({R_1},D,\phi ) = \frac{{{P_{r,1}}}}{{{P_{r,2}} + N}} \geqslant {\gamma _{{\text{th}}}}.
\end{equation}

Obviously, for given $D$ and ${R_1}$ values, a specific range for $\phi$ which satisfies the above inequality, is obtained. That is, for $D = {D_o}$
  and $R = {R_o}$, the positive coverage angle range is ${\phi _{\max }} \leqslant \pi$. Note that since the target area is limited, the effective coverage angle is inside the target area. Hence, the upper bound of  $\phi$ is not necessarily $\pi $ and is replaced by ${\phi _{\max }}$. Assume that the maximum coverage for a DSC in the absence of interference is ${R_{\text{m}}}$. Without loss of generality, we fix the location for ${{\text{DSC}}_{\text{2}}}$ at ${x_2} = \frac{a}{2} - {R_{\text{m}}}$. In addition, for simplicity, we assume that DSCs altitude and their transmit power are fixed and identical and the only parameter that can change is the distance between DSCs. The goal is to find the optimal distance between two DSCs which leads to the maximum overall coverage inside the target area. Note that we fix the position of one DSC over the target area and then deploy the other DSC within distance $D$ from the first one. The total coverage area is expressed as:

\begin{equation}
 \begin{aligned}
{A_C} = {A_{{\text{C}},1}} + {A_{{\text{C}},2}} = 2.\int\limits_{R = 0}^{{R_m}} {{\text{  }}\int\limits_{\phi  = {\phi _{min}}(R)}^{\phi  = {\phi _{max}}(R)} {R.dRd\phi } }  \\
 +2.\int\limits_{R = 0}^{{R_m}} {{\text{  }}\int\limits_{\phi  = {\phi _{min}}(R)}^{\phi  = \pi } {R.dRd\phi } },
 \end{aligned}
\end{equation}
where ${A_{{\text{C}},1}}$ and ${A_{{\text{C}},2}}$  are the effective coverage inside the target area provided by ${{\text{DSC}}_{\text{1}}}$ and ${{\text{DSC}}_{\text{2}}}$. It can be shown that  for ${\phi _{\max }}$ ${{\text{DSC}}_{\text{1}}}$ that might partially cover outside the target area is computed as follows:\vspace{-0.2cm}

\begin{equation}
{\phi _{{\text{max}}}}(R) = {\cos ^{ - 1}}(\max\{  - 1,\frac{{D + {R_{\text{m}}} - a}}{R}\}).
\end{equation}

Finally, the optimal distance between DSCs is \vspace{-0.1cm}

\begin{equation}
{D_{{\text{opt}}}} = \mathop {\arg \max }\limits_D \{{A_C}(D)\}.
\end{equation}

Note that although most of the analytical results shown in the previous sections have closed form expressions, in the case of two fully interfering DSCs, due to the dependency of SINR on the location of ground user, a closed form expression for the total coverage area cannot be derived. In a more general case, the DSCs can be placed at different heights and consequently they can have different coverage performance (${A_{{\text{C}},1}} \ne {A_{{\text{C}},2}}$). The total covered area can be written as: \vspace{-0.3cm}

\begin{equation}
\begin{aligned}
{A_{\text{C}}} = 2.\int\limits_{R = 0}^{{R_{m,1}}} {{\text{  }}\int\limits_{\phi  = {\phi _{min,1}}(R)}^{\phi  = {\phi _{max,1}}(R)} {R.dRd\phi } } \\
 + 2.\int\limits_{R = 0}^{{R_{m,2}}} {{\text{  }}\int\limits_{\phi  = {\phi _{min,2}}(R)}^{\phi  = \pi } {R.dRd\phi } },
\end{aligned}
\end{equation}
where  ${R_{m,1}}$ and ${R_{m,2}}$ is the maximum coverage for the first and second DSCs in the absence of interference,
${\phi _{\min ,1}(R)}$ and ${\phi _{\min ,2}(R)}$ are the minimum angle that for a given $R$ can be covered by DSCs.

In this case, beyond the optimal DSCs distance, the optimal altitudes should also be determined. To this end, a three dimensional search over $D$, $h_1$ and $h_2$ is required. Then we should have \vspace{-0.4cm}

\begin{equation}
({D_{{\text{opt}}}},{h_{{\text{1,opt}}}},{h_{{\text{2,opt}}}}) = \mathop {\arg \max }\limits_{D,{h_1},{h_2}} \{ {A_C}(D,{h_1},{h_2})\}.
\end{equation}

\section{Numerical Results}\label{sec:Results}\vspace{0.1cm}

Assuming that DSCs are operating in urban environments, numerical and analytical results are presented. Table I lists the typical parameters used in the numerical analysis \cite{R3}. Note that the values of $\alpha$ and $\beta$ in (\ref{LOS}) depend on the environment and are different when DSCs operate in other areas such as dense urban or suburban. Here, we consider an urban area and use the corresponding $\alpha$ and $\beta$ parameters to compute the path loss effect.\vspace{-0.2cm}\\

\ Figure 4 shows the minimum transmit power required to have a certain coverage radius as a function of DSC altitude. Deploying a DSC at the optimal altitude minimizes the minimum required transmit power for covering a target area. In fact, for very low altitudes, due to the shadowing impact, the probability of LOS connections between transmitter and receiver decreases and consequently the coverage radius decreases. On the other hand, in very high altitude LOS links exist with a high probability. However, due to the large distance between transmitter and receiver, the path loss increases and consequently the coverage performance decreases. For instance, the optimal altitude for providing 500 m coverage radius while consuming minimum transmit power is 310 m. Moreover, in Figure 4, we can see that only one local minimum exists for the transmit power as a function of altitude. The results in Figure 4 provide very useful guidelines for power minimization which is one of the main concerns in designing DSC networks. Figure 4 shows that as the radius of target area increases, both the optimal altitude and the minimum transmit power required to cover the area increase.

\begin{figure}[!h]
  \begin{center}
   \vspace{-0.02cm}
    \includegraphics[width=8cm]{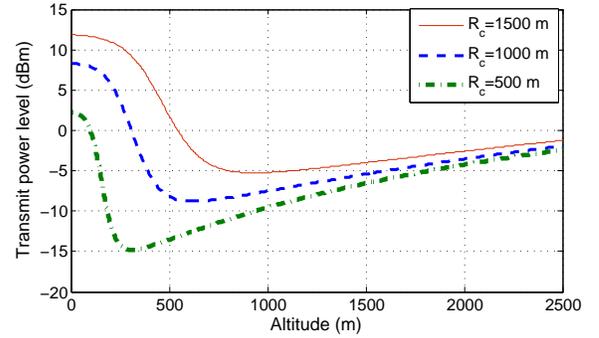}
    \vspace{-0.14cm}
    \caption{\label{fig. H-vs-Pt}  Minimum required transmit power \vspace{0.1cm}}
  \end{center}\vspace{-0.13cm}
\end{figure}

\ In Figure 5, we show the impact of interference on the coverage performance when two DSC are located at an altitude of $300{\text{ m}}$ and a separation distance of $1100$ m. The target area is a rectangle with $a = 2000$ m, $b = 700$ m. The overall coverage area includes two parts inside and outside of the target area.  Note that the effective coverage area is the part of coverage region inside the target area. Figure 5 also shows the impact of interference between DSCs that creates holes between the coverage regions provided by the two DSCs. To maximize the effective coverage area, the distance between two DSCs should be properly adjusted such that the interference between DSCs is not high while the coverage region outside the target area is minimized.
\begin{table}[t]
\captionof{table}{Parameters in numerical analysis}
\begin{center}
\begin{tabular}[b]{|c|c|} %[!t]%to 0.4\textwidth { | X[l] | X[c] |}
%\caption{Parameters in numerical analysis}
 \hline
 \textbf{Parameters} &  \textbf{ Value}\\
 \hline
${f_{\text{c}}}$  & 2 GHz  \\
\hline
${\xi _{{\text{LoS}}}}$  & 1 dB  \\
\hline
${\xi _{{\text{LoS}}}}$  & 20 dB  \\
\hline
$N{\text{(200 KHz bandwith)}}$ & -120 dBm  \\
\hline
$\alpha$ & 9.6  \\
\hline
$\beta$ & 0.28  \\
\hline
length of area ($a$)  & 2000 m  \\
\hline
${\gamma _{{\text{th}}}}$  & 10 dB  \\
\hline
\end{tabular}\vspace{-0.8cm}
\end{center}
\end{table}

\ Figure 6 shows the ratio of  effective coverage area  to the target area that can be achieved using two DSCs for different values of $D$. In the presence of interference for high values of $D$, although the DSCs septation is sufficient to mitigate the impact of interference, they mainly provide coverage for outside of the target area which is not desirable. On the other hand, if the DSCs are very close together, interference between them will significantly reduce the overall coverage performance. As shown in Figure 6, an optimal separation distance between the two DSCs resulting in a maximum coverage in both interference and non-interference cases exists and is about 1100 m and 900 m respectively. In the non-interference situation, as expected, the overall coverage is higher and the optimal separation distance is lower compared to that of in the interference case. The reason is that when there is no interference, we can reduce the DSCs separation distance without loosing the coverage performance that can occur in the presence of interference.

\ In Figure 7, we show the optimal DSCs separation distance as a function of length of the target area. According to Figure 7, the optimal distance between DSCs almost linearly increases according to the size of the area. For example, when the length of the target area changes from 1800 m to 2400 m, the optimal distance between DSCs increases from 1000 m to 1350 m. In fact, to avoid interference between DSCs we should deploy them as separate as possible but still inside the target area. This can be interpreted as scaling the distance between DSCs along with the target area.
%\vspace{-0.1cm}

\begin{figure}[!t]
  \begin{center}
   \vspace{-0.2cm}
    \includegraphics[width=8cm]{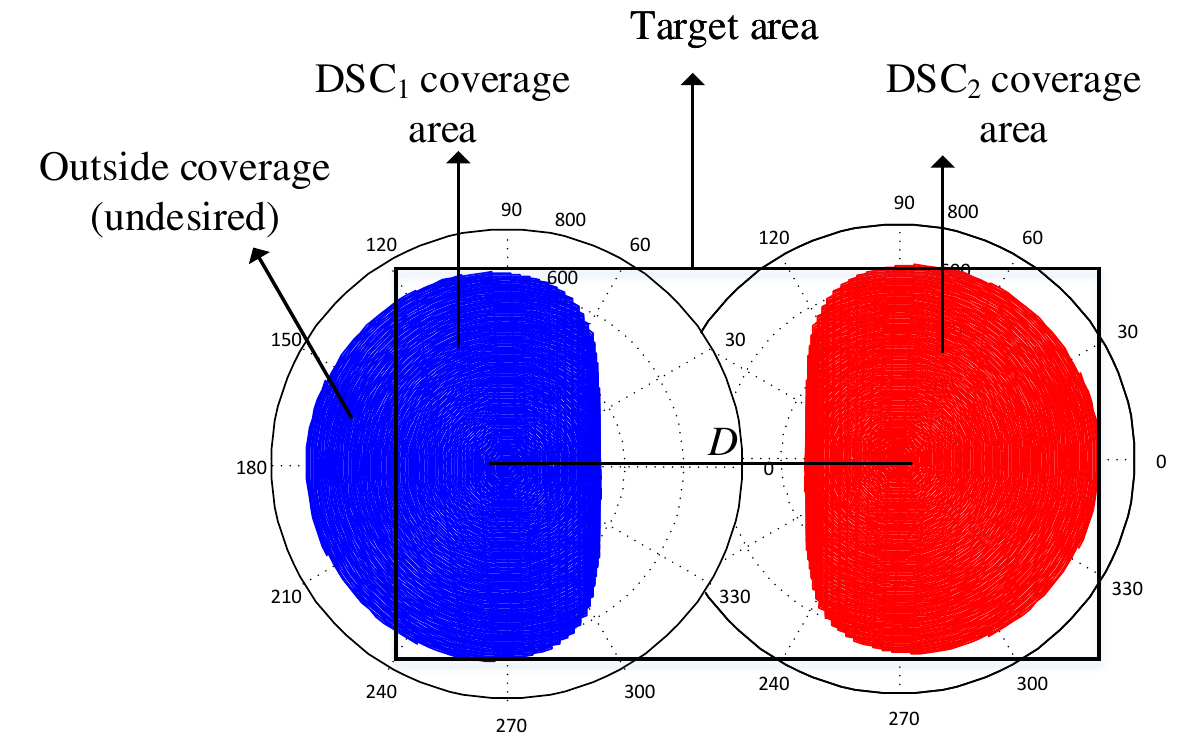}
    \vspace{-0.2cm}
    \caption{\label{fig.Inte0rference}  Coverage performance of two DSCs in the presence of interference.\vspace{-0.02cm}}
  \end{center}\vspace{-0.2cm}
\end{figure}

\begin{figure}[!t]
  \begin{center}
   \vspace{-0.02cm}
    \includegraphics[width=7cm]{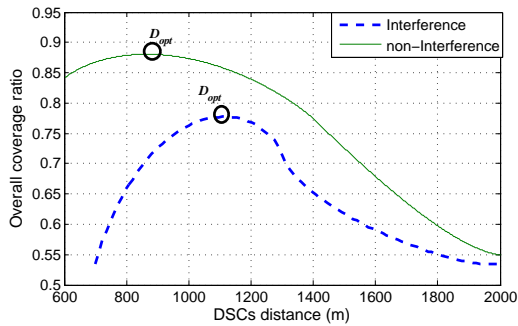}
    \vspace{-0.1cm}
    \caption{\label{fig. D_opt_h=300} Overall coverage ratio versus DSCs separation distance.\vspace{-0.3cm}}
  \end{center}\vspace{-0.5cm}
\end{figure}

\begin{figure}[!t]
  \begin{center}
   \vspace{-0.01cm}
    \includegraphics[width=8.5cm]{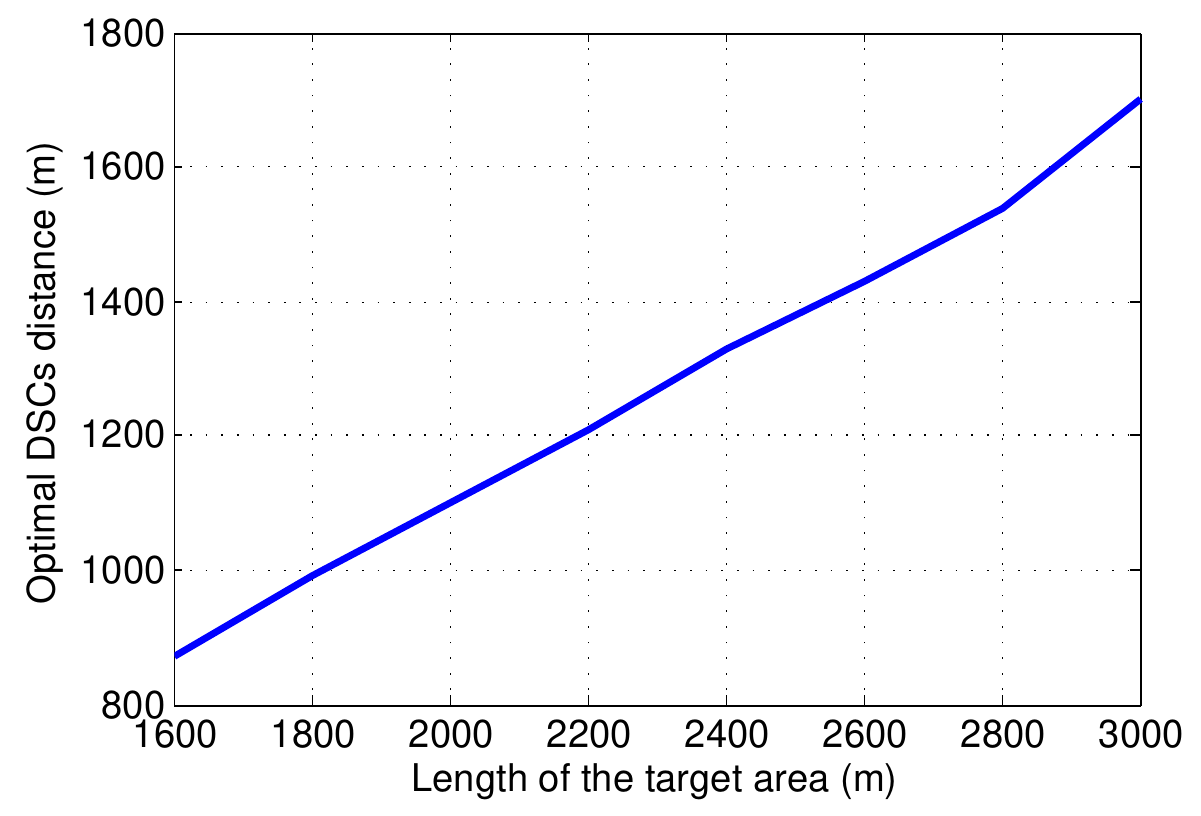}
    \vspace{-0.13cm}
    \caption{\label{fig. D_opt_h=300} Optimal DSCs distance versus length of target area.\vspace{-0.5cm}}
  \end{center}\vspace{-0.4cm}
\end{figure}

\section{Conclusions}\label{sec:Results}\vspace{-0.15cm}

\ In this paper, we have studied the coverage performance of DSCs acting as base stations in low altitude platform. First, the impact of a DSC altitude on the downlink ground coverage has been evaluated and the optimal values for altitude which lead to maximum coverage and minimum required transmit power have been determined. Next, considering an interference free situation and given a target area to be covered, the optimal deployment for two DSCs in terms of altitude and distance between them has been presented. In the presence of full interference between the two DSCs, the coverage area has been formulated. The results have shown the existence of an optimal DSCs separation distance which provides maximum coverage for a given target area. The results presented in the paper provide a stepping stone addressing the more general cases with higher number of DSCs.\vspace{-0.35cm}

%\clearpage

\def\baselinestretch{1.2}
\bibliographystyle{IEEEtran}
\bibliography{references}

% that's all folks
\end{document}